\documentclass{elsarticle}
\usepackage{lineno, hyperref}
\modulolinenumbers[5]
\usepackage{mathtools}

\usepackage{natbib}
\usepackage{fullpage}
\usepackage{setspace}
\usepackage{amssymb,amsthm,latexsym,amsmath,epsfig,pgf}

\usepackage[hang]{subfigure} 
\usepackage{tikz}

\usepackage[noend]{algpseudocode}
\usepackage{graphics}

\newtheorem*{theorem A}{Theorem A}
\newtheorem*{theorem B}{N\"olker's Theorem}

\newtheorem{proposition}{Proposition}[section]

\newtheorem{Defn}{Definition}
\theoremstyle{remark}
\newtheorem{remark}{Remark}[section]
\theoremstyle{remark}
\newtheorem{example}{Example}[section]

\makeatletter
\def\ps@pprintTitle{%
	\let\@oddhead\@empty
	\let\@evenhead\@empty
	\let\@oddfoot\@empty
	\let\@evenfoot\@oddfoot
}
\makeatother
\begin{document}
\begin{frontmatter}
	\title{An Unique and Novel Graph Matrix for Efficient Extraction of Structural Information of Networks} %,t2}}
	
	%\tnotetext[t2] {}
	\author[siva]{Sivakumar Karunakaran} %\fnref{fn1}}
	\ead{sivakumar\_karunakaranm@srmuniv.edu.in}
	\author[lavanya]{Lavanya Selvaganesh\corref{cor1}} % \fnref{fn2}}
	\ead{lavanyas.mat@iitbhu.ac.in}
	\cortext[cor1]{Corresponding author.}
	%\fntext[fn1]{}
	%\fntext[fn2]{}
	
	\address[siva]{SRM Research Institute, S R M Institute of Science and Technology Kattankulathur, Chennai - 603203, INDIA}
	
	\address[lavanya]{Department of Mathematical Sciences, Indian Institute of Technology (BHU), Varanasi-221005, INDIA}
	
\begin{abstract}
	In this article, we propose a new type of square matrix associated with an undirected graph by trading off the naturally imbedded symmetry in them. The proposed matrix is defined using the neighbourhood sets of the vertices. It is called as neighbourhood matrix and it is denoted by $ \mathcal{NM}(G)$ as this proposed matrix also exhibits a  bijection between the product of the two graph matrices, namely the adjacency matrix and the graph Laplacian. This matrix can also be obtained by looking at every vertex and the subgraph with vertices from the first two levels in the level decomposition from that vertex. The two levels in the level decomposition of the graph gives us more information about the neighbour of a vertex along with the neighbour of neighbour of a vertex. This insight is required and is found useful in studying the impact of broadcasting in social networks, in particular, and complex networks, in general. We establish several interesting properties of the $ \mathcal{NM}(G) $. In addition, we also show how to reconstruct a graph $G$, given an $ \mathcal{NM} (G)$. The proposed  matrix is also found to solve many graph theoretic problems using less time complexity in comparison to the existing algorithms.\footnote{A preliminary version of this article namely the definition of the newly proposed matrix was presented in ICDM 2016(June 09-11), Siddaganga Institute of Technology, Tumkur-572102, Karnataka, INDIA and few of the characterizations were presented in the Fifth India-Taiwan Conference on Discrete Mathematics(18-21 July, 2017), Tamkang University, Taiwan.}

\end{abstract}


\end{frontmatter}

\section{Introduction}
In the study of complex and social networks, one of the interesting and challenging problem is to study the impact of a change that occurs to a node. Such studies are being done to analyse the network's behavioural changes both locally as well as globally, \cite{BIOLOGY}.  One such problem is in reconstructing a graph when partial information is known and/or predict the dynamical changes occurring in a network. To tackle this problem, we were determined to approach it by studying graphs through their matrices.

Matrices play an important role in the study of graphs and their representations. It is well known that for undirected graphs, among all graph matrix forms, adjacency matrix and Laplacian matrix has received wide attention due to their symmetric nature \cite{GraphsMatrices, Graphtheory, MCM3}. In the literature, many other types of matrices that could be associated with a graph \cite{GraphsMatrices, Bapatetal, MCM3, SURVEY}. For an undirected graph, every such matrix is found to be symmetric and is not of help to solve our problem. Further, in \cite{Bapatetal}, the authors discuss about the product of two graphs and its representation using product of the adjacency matrices of the graphs. However, there is no literature dealing with the product of two types of matrices of a graph.

In this paper, we handle one such problem involved in defining, analysing and correlating the product of graph matrices with the graph and several of its properties. To this end, we propose a novel representative matrix for a graph referred to as $ \mathcal{NM}(G)$. We first define this matrix by using the notion of neighbourhood of a vertex in a graph and then endorse its relationship with the product of  two different types of graph matrices. We make sure that the matrix that we are defining in this paper is not always symmetric and this helps us in proving many network properties quite easily.

The paper is organized as follows: In section 2, we present all the basic definitions, notations and properties required. In subsection of 2, we introduce the novel concept of $ \mathcal{NM}(G) $ and discuss several of its properties. In section 3, we discover some interesting characterizations of the graph using the $ \mathcal{NM}(G) $. We conclude the paper in section 4 with some insight on future scope.

\section{ Definitions and Notations:}
Throughout this paper, we consider only undirected, unweighed simple graphs. For all basic notations and definitions of graph theory, we follow the books by J.A. Bondy and U.S.R. Murty \cite{Graphtheory} and D.B. West \cite{GraphTheoryWest}. In this section, we present all the required notations and define the $ \mathcal{NM} (G)$. Let $ G(V,E) $ be a graph with vertex set $ V(G) $ and edge set $ E(G) $. For a vertex $v\in V(G)$, let $N_{G}(v)$ denote the set of all neighbours of v and  $N_{G}[v]=\{v\}\cup N_{G}(v)$, denote the closed neighbourhood of $ v $. The degree of a vertex $v$ is given by $ deg(v) $ or $|N_{G}(v)|$. Let $ A_{G} $ (or $ A $) denote the adjacency matrix of the grpah $ G $. Let the degree matrix $ D(G) $ (or $ D $ ) be the diagonal matrix with the degree of the vertices as its diagonal elements. Let $ C(G) $ be the Laplacian matrix obtained by $ C(G) = D(G)-A_{G} $.
\begin{Defn}
	Given a graph $G$, the product of the adjacency matrix and the degree matrix, denoted by $AD=[ad_{ij}]$, is defined as $$ad_{ij}=\begin{cases}
	|N_{G}(j)|,&    \text{ if }   (i,j)\in E(G)\\
	0 , &  otherwise \\
	\end{cases} $$\\
	Similarly, the product of the degree matrix and the adjacency matrix, denoted by $ DA=[da_{ij}] $, is defined as $$da_{ij}=\begin{cases}
	|N_{G}(i)|,&    \text{ if }   (i,j)\in E(G)\\ 
	0,  &  otherwise\\
	\end{cases} $$
\end{Defn}
\begin{remark}
	From the above definitions it follows immediately that $AD^T=DA.$
\end{remark}
\begin{remark}\label{Remark.2}
	If  $G$  is regular or contains regular-components then by the definition, $AD$ matrix is symmetric. Hence by above remark $AD$ and $DA$ becomes equal.
\end{remark}
\begin{Defn}
	Given a graph $ G $, the square of the adjacency matrix $A^2=[a^2_{ij}]$, is defined as 
	$$a^2_{ij}=\begin{cases}
	|N_{G}(i)|,  &  \text{ if }    i=j\\
	|N_{G}(i)\cap N_{G}(j)|,  &    \text{ if }  i\neq j  \\
	\end{cases} $$
\end{Defn}
It is well known that the $ ij ^{th}$ entries of the square of adjacency matrix denotes the number of walks of length 2 between $ i $ and $ j $. 

Another concept which we require before proceeding to the main result is the level decomposition of a graph with respect to a source node, which is defined by the Breadth First Search Traversal technique. 

\textit{Breadth First Search (BFS)} is a graph traversal technique \cite{Graphtheory} where a node (source node) and its neighbours are visited first and then the neighbours of neighbours. The algorithms returns not only a search tree rooted at the source node but also a function $l: V\rightarrow \mathbb{N}$, which records the level of each vertex in the tree, that is, the distance of each vertex from the source node. In simple terms, the BFS algorithm traverses level wise from the source. First it traverses level 1 nodes (direct neighbours of source node) and then level 2 nodes (neighbours of neighbours of source node) and so on. We refer to such a level representation with reference to a source node as the level decomposition from the source node.  

We next extend the above notion of product of graph matrices to obtain a new class of matrix and establish its properties.

\subsection{$\mathcal{NM}(G)$ and its properties}
Now we introduce the idea of $ \mathcal{NM}(G) $ and describe its properties 
\begin{Defn}
	Given a graph $ G $, the neighbourhood matrix, denoted by $ \mathcal{NM}(G) =[\eta_{ij}]$ is defined as 
	$$\eta_{ij}=\begin{cases}
	-|N_{G}(i)|,  & \text{ if }  i=j\\
	|N_{G}(j)-N_{G}(i)|,  &   \text{ if }   (i,j)\in E(G)\\
	-|N_{G}(i) \cap N_{G}(j)|,  & \text{ if }    (i,j)\notin E(G)\\
	\end{cases} $$ 	
\end{Defn}
\begin{example}\label{Exap}
	A graph $ G $ and its corresponding $ \mathcal{NM} (G)$ representation are given in \normalfont{Figure}~\ref{Figure.1}. In this example, the neighbourhood set of each vertex of $ G $ is given by $N_{G}(1)$= \{2,6\}, 
	$N_{G}(2)$=\{1,5\}, $N_{G}(3)$=\{4\}, $N_{G}(4)$=\{3,5\}, $N_{G}(5)$=\{2,4,6,7\}, $N_{G}(6)$=\{1,5,7\}, $N_{G}(7)$=\{5,6\}.
	
	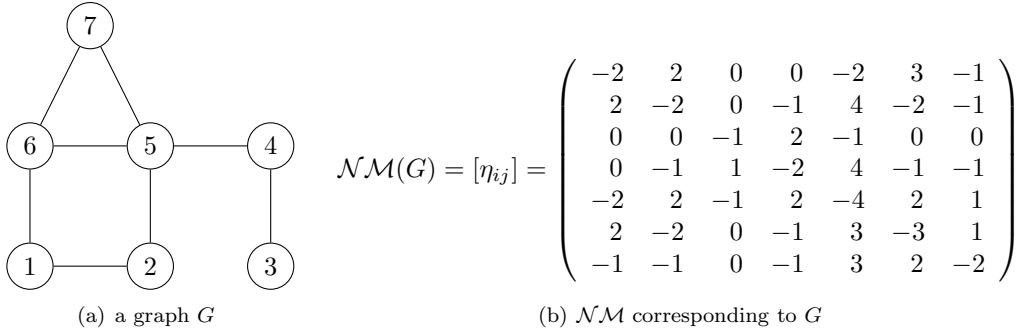
\begin{figure}[ht!]
		\centering 
		\subfigure[a graph $ G $]{
			\begin{tikzpicture}[scale=0.8,auto=left,every node/.style={draw,circle}] 
			\node (n1) at (0,0)  {1};
			\node (n2) at (2,0)  {2};  
			\node (n3) at (4,0)  {3};
			\node (n4) at (4,2)  {4};
			\node (n5) at (2,2)  {5};
			\node (n6) at (0,2)  {6};
			\node (n7) at (1,4)  {7};
			\foreach \from/\to in {n1/n2, n1/n6,n2/n5,n3/n4, n4/n5, n5/n6, n5/n7, n6/n7}
			\draw (\from) -- (\to)      ;
			\label{Figure.1a}  \end{tikzpicture}} 
		\quad  \subfigure[$ \mathcal{NM} $ corresponding to $G$]
		{\begin{tikzpicture} \label{Figure.1b}  \node (0,2){  $ \mathcal{NM}(G)  = [\eta_{ij}]=\left( \begin{array}{c c c c c c c}  
				-2&     ~~2&     ~~0&     ~~0&    -2&     ~~3&    -1\\
				~~2&    -2&     ~~0&    -1&     ~~4&    -2&        -1\\
				~~0&     ~~0&    -1&     ~~2&    -1&     ~~0& ~~0\\
				~~0&    -1&     ~~1&    -2&     ~~4&    -1&      -1\\
				-2&     ~~2&    -1&     ~~2&    -4&     ~~2&     ~~1\\
				~~2&    -2&     ~~0&    -1&     ~~3&    -3&     ~~1\\
				-1&    -1&     ~~0&    -1&     ~~3&     ~~2&    -2\\
				\end{array} \right) $}; \end{tikzpicture}} 
		\caption{A graph $G$ and its $ \mathcal{NM}(G) $.}
		\label{Figure.1}
	\end{figure}
\end{example}
%\ref{Figure.1a} \ref{Exap} \ref{Figure.1b}
\begin{proposition}\label{prop.1}
	The $ \mathcal{NM}(G) $ can also be defined by using the product of adjacency matrix and Laplacian matrix of a graph $ G $. 
\end{proposition}
\begin{proof} Consider the definition of product of two matrices 
	\begin{eqnarray*}
		A \times C(G) &= & A \times (D(G)-A)\\
		&   =&  AD-A^2\\
		&   =& [ad_{ij}]-[a^2_{ij}]\\		
		&   = &\begin{cases}
			0-|N_{G}(i)|,  &  \text{ if }  i=j\\
			|N_{G}(j)-N_{G}(i)|,  &    \text{ if }  (i,j)\in E(G)\\
			0-|N_{G}(i) \cap N_{G}(j)|,  & \text{ if }    (i,j)\notin E(G)\\
		\end{cases}
	\end{eqnarray*} 
	Note that the last equality represents the  $\mathcal{NM}(G)$. Hence the proof. \end{proof}

\begin{proposition}
	Given a graph $ G $, the $ \mathcal{NM}(G) $ can be obtained from adjacency matrix and vice versa. 
\end{proposition}
\begin{proof} By \textnormal{Proposition \ref{prop.1}}, it is immediate that the matrix $ \mathcal{NM}(G) $ can be constructed from the adjacency matrix. 
	
	Given the $ \mathcal{NM}(G) $, if $ i\neq j, \eta_{ij}>0 $ implies that by the definition of $ \eta_{ij}=|N_{G}(j)-N_{G}(i)| $ for $ (i,j)\in E(G)   $ and if $ \eta_{ij}\leq 0 $ this implies either $ i=j $ or $ (i,j)\notin E(G) $.
	
	Therefore, we can now define $ a_{ij}=\begin{cases}
	1,  &  \text{ if } \eta_{ij}>0\\
	0, & Otherwise
	\end{cases} $
\end{proof}

\begin{example}
	From the $ \mathcal{NM}(G) $ in {Figure \ref{Figure.1b}}, constructing the adjacency matrix as defined in the above proposition, we get,
	\begin{figure}[ht!] \scriptsize 
		\centering 
		$  A = [a_{ij}]=\left( \begin{array}{c c c c c c c}          
		0&     ~~1&     ~~0&     ~~0&     ~~0&     ~~1&     ~~0\\
		1&     ~~0&     ~~0&     ~~0&     ~~1&     ~~0&     ~~0\\
		0&     ~~0&     ~~0&     ~~1&     ~~0&     ~~0&     ~~0\\
		0&     ~~0&     ~~1&     ~~0&     ~~1&     ~~0&     ~~0\\
		0&     ~~1&     ~~0&     ~~1&     ~~0&     ~~1&     ~~1\\
		1&     ~~0&     ~~0&     ~~0&     ~~1&     ~~0&     ~~1\\
		0&     ~~0&     ~~0&     ~~0&     ~~1&     ~~1&     ~~0 \\
		\end{array} \right) $
		\caption{Adjacency matrix of $G$ constructed from $\mathcal{NM}(G)$}\label{Figure.3}
	\end{figure} 
	
	It is immediate that $ A $ is the required adjacency matrix.
\end{example}

An alternative interpretation or a way of defining the $ \mathcal{NM}(G) $ is to consider the breadth first traversal starting at a vertex $i$. By inspection of the first two levels in this level decomposition, we can obtain the respective $ i^{th} $ row of the $ \mathcal{NM}(G) $. We prove this equivalence in the following proposition.

\begin{proposition}\label{prop.2}
	Given a graph $ G $, the entries of any row of an $ \mathcal{NM}(G)$ corresponds to the subgraph with vertices from the first two levels of level decomposition of the graph rooted at the given vertex with edges connecting the vertices in different levels.
\end{proposition} 

\begin{proof}
	Consider any $i ^{th} $ row of the $ \mathcal{NM}(G)$. By the definition of $ \mathcal{NM}(G) $, vertex $ i $ is adjacent to a vertex $ j \iff \eta_{ij}>0 $. This gives us the neighbours of $ i $, namely $ N_{G}(i) $, or the first level of the level decomposition. From the following observations, we obtain the vertices that lie in the next level. 
	\begin{enumerate}
		\item The diagonal entries are always negative and in particular, if $ \eta_{ii}=-c $, then the degree of the vertex is $ c $ and that there will be exactly $ c $ positive entries in that row.
		\item For some positive integer $ c $, if $ \eta_{ij}=c $ then $ j\in N_{G}(i) $ and that there exists $ c-1 $ vertices are adjacent to $ i $ and at distance $ 2 $ from $ i $ through $ j $.
		\item If $ \eta_{ij}=-c $, then the vertex $ j $ belongs to the second level of the decomposition and moreover, there exists $ c $ paths of length two from vertex $ i $ to $ j $. In other words, there exist $ c $ common neighbours between vertex $ i $ and $ j $.
		\item If an entry, $ \eta_{ij}=0 $ then the distance between vertex $ i $ and $ j $ is at least $ 3 $ or the vertex $ j $ is isolated
	\end{enumerate}
	Combining these observations, one can easily obtain the subgraph with vertices from the first two levels of decomposition of $  G  $ rooted at the vertex $ i $.\\
	On the other hand, from the Breadth first traversal tree rooted at a vertex $ i $ and the definition of $ \mathcal{NM}(G) $ we can immediately write the corresponding $ i^{th} $ row entry by examining the vertices and their position in the first two levels.
\end{proof}

%\begin{example}
%From the \textnormal{Figure \ref{Figure.1b}}  we obtain the following
% \begin{enumerate}
%\item $\eta_{11}=-2$ denotes there exist two neighbours of vertex $ 1 $.
%\item $\eta_{12}=2$(positive entry) denotes there is an adjacency between vertex $ 1 $ and vertex $ 2 $ and there exist one vertex (vertex $ 5 $) at distance $ 2 $ from vertex $ 1 $ through vertex $ 2 $.
%\item $ \eta_{13}=0 $ denotes that vertices $ 1 $ and $ 3 $ are at distance at least $ 3 $. 
%\item $\eta_{15}=-2$ (negative entry) denotes that there exist no adjacency between vertex $ 1 $ and $ 5 $ and there exist one vertex which is a common neighbour of vertex $ 1 $ and vertex $ 5 $. 
%\end{enumerate}
%\end{example}
%

Analogous to $ \mathcal{NM}(G) $ we can also define the product matrix $ \mathcal{MN}(G) $ as follows.

\begin{Defn}
	The product of Laplacian matrix and adjacency matrix denoted by $ \mathcal{MN}=[\eta'_{ij}] $ is defined as 
	$$\eta'_{ij}=\begin{cases}
	-|N_{G}(i)| , &  \text{ if }  i=j\\
	|N_{G}(i)-N_{G}(j)|,  &    \text{ if }    (i,j)\in E(G)\\
	-|N_{G}(i) \cap N_{G}(j)| , & \text{ if }    (i,j)\notin E(G)\\
	\end{cases} $$ \\
\end{Defn}

\begin{remark}
	Note that $ \mathcal{MN}(G) $ can be obtained by $ C\times A=DA-A^{2} $. 
\end{remark}

\begin{remark}\label{Remark.1}
	For an undirected simple graph \begin{eqnarray*}
		( \mathcal{NM} )'	&=& (A\times C)'\\
		&=& C' \times A' \\
		&=& C\times A\\
		&=& \mathcal{MN}
	\end{eqnarray*}
\end{remark}

\begin{proposition}
	The $ \mathcal{NM} $ matrix is a singular matrix 
\end{proposition}
\begin{proof} 
	Let $ A $ be an adjacency matrix and $ C(G) $ be the Laplacian  matrix. 
	
	It is enough to prove $det( \mathcal{NM} )=0$. Since \begin{eqnarray*}
		det( \mathcal{NM} )&= & det(A\times C)\\
		& = & det(A)\times det(C)\\
		&=& 0.
	\end{eqnarray*} 
	Since it is well know that, $det(C)=0$ we get the last equality and hence the claim. 
\end{proof}

\begin{proposition} \label{Row sum zero}
	Row sum of $ \mathcal{NM}(G) $ is zero.
\end{proposition}

\begin{proof} Consider any $i^{th}$ row in $ \mathcal{NM}(G) $
	\begin{equation}
	\sum\limits_{j=1}^n \eta_{ij}=\sum\limits _{j\in N_{G}(i)}|N_{G}(j)-N_{G}(i)|-|N_{G}(i)|-\sum\limits_{j\notin N_{G}[i]}|N_{G}(i) \cap N_{G}(j)|
	\label{eqno}
	\end{equation}
	
	\begin{equation}\label{Row sum EQ.1}
	\sum\limits_{j=1}^n \eta_{ij} =\sum\limits_{j\in N_{G}(i)}(|N_{G}(j)-N_{G}(i)|-1) -\sum\limits_{j\notin N_{G}[i]} |N_{G}(i)\cap N_{G}(j)|
	\end{equation}
	Consider the level decomposition of the graph $ G $ from the vertex $ i $.
	
	Observe that, $\sum\limits_{j\in N_{G}(i)}(|N_{G}(j)-N_{G}(i)|-1)$ is the number of edges connecting the vertices from level $ 1 $ to level $ 2 $. Similarly, $\sum\limits_{j\notin N_{G}[i]}|N_{G}(i)\cap N_{G}(j)|$ denote the edges connecting the vertices from level $ 2 $ to level $ 1 $. So, we have
	
	\begin{equation} \label{Row sum EQ.2}
	\sum\limits_{j\in N_{G}(i)}(|N_{G}(j)-N_{G}(i)|-1)=\sum\limits_{j\notin N_{G}[i]}|N_{G}(i)\cap N_{G}(j)|
	\end{equation}
	Substitute the {equation} (\ref{Row sum EQ.2}) in {equation} (\ref{Row sum EQ.1}) we get the row sum of $ \mathcal{NM}(G) $ is zero.
\end{proof}

\begin{remark}
	Suppose any row of $ \mathcal{NM}(G) $ is given, the degree of the vertex the row represents can be obtained from the minimum value of that row. By considering this position as the diagonal position of the row (since $\eta_{ii} = -|N_{G}(i)|$), hence enables us to identify the vertex that it represent. 
\end{remark}

\begin{example} \label{Exap.2} Consider the row given by \textnormal{Figure \ref{Figure4a}} from the \textnormal{Example \ref{Exap}}, we see that the minimum value is $-4$ occurring at $ 5^{th} $ position of the row, tells us that the row represents vertex $ 5 $ in the example.
	\begin{figure}[ht!] 
		\centering 
		\subfigure[$i^{th} $ row of $ \mathcal{NM} $ matrix]{\label{Figure4a}
			$  	[a_{ij}]=\left( \begin{array}{c c c c c c c}             -2&     ~~2&    -1&     ~~2&    -4&     ~~2&     1     
			\end{array} \right) $ } 
		\quad  \subfigure[A subgraph corresponding to the row matrix ]
		{\label{Figure4b} \begin{tikzpicture}[scale=0.4,auto=left,every node/.style={draw,circle}]
			[scale=.25,auto=left,every node/.style={draw,circle}]           
			\node (n1) at (2,2)  {1};
			\node (n2) at (8,6)  {2};  
			\node (n3) at (7,2)  {3};
			\node (n4) at (6,6)  {4};
			\node (n5) at (5,10)  {5};
			\node (n6) at (4,6)  {6};
			\node (n7) at (2,6)  {7};
			
			\foreach \from/\to in {n5/n2,n5/n4,n5/n6,n5/n7,n2/n1,n4/n3,n6/n1}  
			\draw (\from) -- (\to);    
			`            \end{tikzpicture}} 
		\caption{ Row matrix of $ \mathcal{NM}(G) $ get from Example \ref{Exap} and its graph representation  }\label{Figure4}   
	\end{figure}
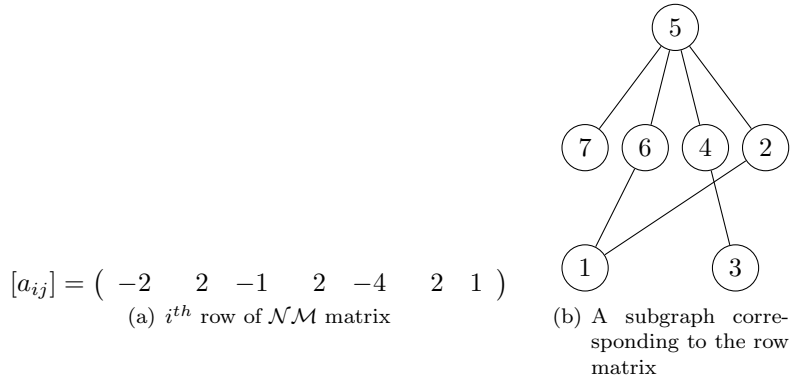    
	In addition, using the row entries and the two level decomposition, we can construct the subgraph with vertices from the first two levels of the level decomposition of the graph rooted at vertex $i$. The  \textnormal{Figure \ref{Figure4b}} shows the constructed subgraph rooted at vertex 5 by using the corresponding row entries. Further from \textnormal{Figure \ref{Figure4a}}, we also get the row sum of $ \mathcal{NM}(G) $ is zero. 
\end{example}  

\begin{proposition}
	For any $ 1\leq i\leq n $, the $ i^{th} $ column sum  of $ \mathcal{NM}(G) $ is equal to  
	
	$\sum \limits_{j\in N_{G}(i)}\Bigg(|N_{G}(i)|-|N_{G}(j)|\Bigg).$
\end{proposition}
\begin{proof}
	By Remark \ref{Remark.1}, we have $(\mathcal{MN})'= \mathcal{NM} $. This implies the column sum of $ \mathcal{NM} $ matrix is equal to the row sum of $ \mathcal{MN} $ matrix. Therefore, we get
	\begin{eqnarray*}
		\tiny \sum_{j=1}^n \eta_{ji} = \sum_{j=1}^n \eta'_{ij} & = &\sum \limits_{j\in N_{G}(i)}|N_{G}(i)-N_{G}(j)|-\sum \limits_{j\notin N{[i]}}|N_{G}(i) \cap N_{G}(j)|-|N_{G}(i)|\\
		&= &\sum \limits_{j\in N_{G}(i)}|N_{G}(i)|-\sum \limits_{j\in N_{G}(i)}|N_{G}(i) \cap N_{G}(j)|-\sum \limits_{j\notin N{[i]}}|N_{G}(i) \cap N_{G}(j)|-|N_{G}(i)|\\
		%&=&\sum \limits_{j\in N_{G}(i)}|N_{G}(i)|-\sum \limits_{j\in N_{G}(i)}|N_{G}(i) \cap N_{G}(j)|-\Bigg(\sum \limits_{j\in N_{G}(i)}|N_{G}(j)|-\sum \limits_{j\in N_{G}(i)}|N_{G}(i) \cap N_{G}(j)|\Bigg)  (\sc \textnormal{ Proposition }  \ref{Row sum zero})\\
		&=&\sum \limits_{j\in N_{G}(i)}\Bigg( |N_{G}(i)|-|N_{G}(j)|\Bigg)(\textnormal{by Proposition }  \ref{Row sum zero}) 
	\end{eqnarray*}
	Hence the proof. 
\end{proof}

\section{Graph characterization using neighbourhood matrix $ \mathcal{NM}(G) $}
Note that the matrix $ \mathcal{NM} (G)$ is not always symmetric. The next result characterizes the graphs for which $ \mathcal{NM}(G) $ will be symmetric.
\begin{proposition}
	The $ \mathcal{NM}(G) $  is symmetric if and only if the graph $ G $ is either regular or contains regular components.
\end{proposition}
\begin{proof} Let $ G $ be a graph with $ w(G) $ components say $ G_1, G_2, ..., G_w $ such that each $ G_z $ is regular with degree $ r_z ,  1\leq z\leq w(G) $
	By the definition of $ \mathcal{NM}(G) $ when $ i $ is not adjacent to $ j $ then $\eta_{ij}=\eta_{ji}$ and when $ i $ is adjacent to $ j $ then \begin{equation}\label{Regular eq1}
	\eta_{ij}=|N_{G}(j)|-|N_{G}(i)\cap N_{G}(j)|=r_z-|N_{G}(i)\cap N_{G}(j)|
	\end{equation}
	\begin{equation}\label{Regular eq2}
	\eta_{ji}=|N_{G}(i)|-|N_{G}(i)\cap N_{G}(j)|=r_z-|N_{G}(i)\cap N_{G}(j)|
	\end{equation}
	From (\ref{Regular eq1}) and (\ref{Regular eq2})  we have $\eta_{ij}=\eta_{ji}$. 
	Therefore the $ \mathcal{NM}(G) $ is symmetric when the graph $ G $ has regular components. 
	
	Conversely, let the $ \mathcal{NM}(G) $ be symmetric. We know that, $ \mathcal{NM}(G) $ can be written as $AD-A^{2}$.
	Since sum of symmetric matrices is symmetric and $ AD =  \mathcal{NM}  + A^{2} $, we must have $ AD $ to be symmetric.  But from Remark \ref{Remark.2}, it is known that $ AD $ is symmetric whenever $ G $ is the union of regular components.
\end{proof}

Recall that a graph $ G $ is said to be a strongly regular graph with parameters $ (n, k, \mu_{1}  , \mu_{2}) $, if $ G $
is a $ k $-regular  graph on n vertices in which every pair of adjacent vertices has $ \mu_{1} $ common neighbours and every pair of non-adjacent vertices has $ \mu_{2} $ common neighbours. \\

\begin{proposition}\label{Stronglyregular}
	If a graph $ G $ is strongly regular then the entries of $ \mathcal{NM}(G) $ contains either two or three distinct values. 
\end{proposition}
\begin{proof}
	By the definition of $ \mathcal{NM}(G) $ it immediate follows that for a strongly regular graph $G$,  
	$$\eta_{ij}(G)=\begin{cases}
	-k, &  \text{ if } i=j\\
	k-\mu_1,  &    \text{ if }    (i,j)\in E(G)\\
	-\mu_2,  & \text{ if }    (i,j)\notin E(G)\\
	\end{cases} $$ \\
	where $ \mu_1=|N_{G}(i)\cap N_{G}(j)|,$ for $(i,j)\in E(G)$  and $\mu_2=|N_{G}(i)\cap N_{G}(j)|,$ for $(i,j)\notin E(G) $.
	This implies the entries of $ \mathcal{NM}(G) $ of a strongly regular graph takes values from $ \{-k, k-\mu_{1}, -\mu_{2} \}$ or $ \{-k, k-\mu_{1}\}$,  when $k=\mu_{2}  $. 
	% when G is $K_k$ or $K_{k,k}$ 
\end{proof}

\begin{remark}
	Note that the converse of the above proposition need not be true.
\end{remark}

\begin{example}
	\normalfont{Figure} \ref{Figure.5a} is the $ \mathcal{NM}(G) $ containing only three distinct values as entries, namely, $\{-2,0,2\}$.  Figure \ref{Figure.5b} is the corresponding graph of Figure \ref{Figure.5a}. Note that the graph is not a strongly regular graph.
	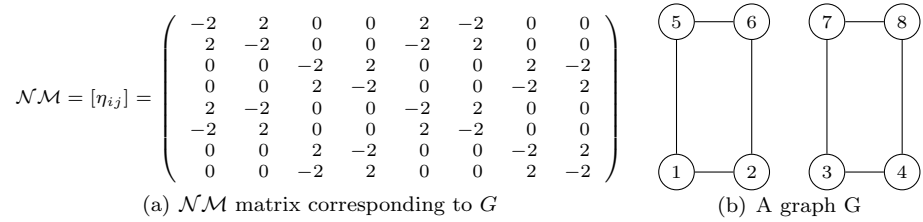
\begin{figure}[ht!] \scriptsize
		\centering 
		\subfigure[$ \mathcal{NM} $ matrix corresponding to $G$]  
		{\begin{tikzpicture} \label{Figure.5a}  \node (0,0.1){  $ \mathcal{NM}  = [\eta_{ij}]=\left( \begin{array}{c c c c c c c c }  
				-2&     ~~2&     ~~0&     ~~0&     ~~2&    -2&     ~~0&     ~~0\\
				~~2&    -2&     ~~0&     ~~0&    -2&     ~~2&     ~~0&     ~~0\\
				~~0&     ~~0&    -2&     ~~2&     ~~0&     ~~0&    ~~2&    -2\\
				~~0&     ~~0&     ~~2&    -2&     ~~0&     ~~0&    -2&     ~~2\\
				~~2&    -2&     ~~0&     ~~0&    -2&     ~~2&     ~~0&     ~~0\\
				-2&     ~~2&     ~~0&     ~~0&     ~~2&    -2&     ~~0&     ~~0\\
				~~0&     ~~0&     ~~2&    -2&     ~~0&     ~~0&    -2&     ~~2\\
				~~0&     ~~0&    -2&     ~~2&     ~~0&     ~~0&     ~~2&    -2\\
				\end{array} \right) $}; \end{tikzpicture}} 
		\quad  \subfigure[A graph G]    
		{\begin{tikzpicture}[scale=1,auto=left,every node/.style={draw,circle}] 
			\node (n1) at (1,0)  {1};
			\node (n2) at (2,0)  {2};  
			\node (n3) at (3,0)  {3};
			\node (n4) at (4,0)  {4};
			\node (n5) at (1,2)  {5};
			\node (n6) at (2,2)  {6};
			\node (n7) at (3,2)  {7};
			\node (n8) at (4,2)  {8};
			\foreach \from/\to in {n1/n2,n5/n1,n6/n2, n3/n4, n3/n7, n4/n8, n7/n8,n5/n6}
			\draw (\from) -- (\to)      ;
			\label{Figure.5b}  \end{tikzpicture}}
		\caption{A graph $G$ and its $ \mathcal{NM} $ matrix.}
		\label{Figure.5}
	\end{figure} 
\end{example}

\begin{proposition}
	If at least one row of $ \mathcal{NM}(G) $ has no zero entries then the graph $ G $ has diameter at most 4.
\end{proposition}
\begin{proof} Let $ i^{th} $ row of $  \mathcal{NM} (G) $ have no zero entries then by using the two level decomposition definition we have $ d_{G}(i,j)\leq 2, \forall  j\in V(G)-i.$, otherwise, $ d_{G}(i,j)=3 $ for some $ j $ this implies $ \eta_{ij}=0 $.\\
	Therefore for any $ j,k\in V(G)-i $ we have $ d_{G}(i,j)\leq 2 $ and $ d_{G}(i,k)\leq 2. $\\
	So, $ d_{G}(j,k)\leq d_{G}(j,i)+d_{G}(i,k)\leq 2+2=4 $ 
\end{proof}
\begin{remark}
	Note that the converse of the above proposition need not be true. It is well known that the cubic graph on 8 verrtices ($ Q_{3} $) has diameter 3 but every row of  $\mathcal{NM} (Q_{3}) $ contains exactly one zero.
\end{remark}

\begin{proposition}
	The $ \mathcal{NM}(G) $ has no zero entries if and only if the graph $ G $ has diameter at most 2.
\end{proposition}
\begin{proof} $ \mathcal{NM}(G) $ has no zero entries $\iff$  for every $i$,  $i^{th}$ row of $ \mathcal{NM}(G) $ has no zero entries 
	$\iff  \forall i,j, i \neq j, d_{G}(i,j)\leq 2,
	\iff diameter(G)\leq 2$.
\end{proof} 
\begin{proposition}\label{Triangle free}
	The graph $ G $ is triangle-free if and only if $\eta_{ij}=|\eta_{jj}|$ $\forall (i,j)\in E(G)$.
\end{proposition}
\begin{proof} Graph $ G  $ is triangle-free $\iff  N_{G}(i)\cap N_{G}(j)=\emptyset $ for $ (i,j)\in E(G)$. By the definition of $ \mathcal{NM}(G) $ if $ i $ is adjacent to $ j $ then $ \eta_{ij}=|N_{G}(j)-N_{G}(i)|=|N_{G}(j)|-|N_{G}(i)\cap N_{G}(j)|=|N_{G}(j)| $. Now in the $ i^{th } $ row $ \eta_{ij} = |N_{G}(j)|=|\eta_{jj}|$.
\end{proof}
\begin{proposition} Given a graph $G$, the number of triangles in $G$ is given by  
	$$\dfrac{1}{6} \sum\limits_{i}\sum\limits_{j\in N_{G}(i)}\Bigg(|\eta_{jj}|-\eta_{ij}\Bigg).$$
\end{proposition}
\begin{proof}
	Given a vertex $ i $, when $ i $ is adjacent to $ j $ and there exists at least one common neighbour $ x $, for $ i $ and $ j $, we get a triangle. 
	
	$ \therefore  $ Number of triangle containing the vertex $i$ is given by $NT(i) = \dfrac{1}{2} \sum\limits_{j\in N_{G}(i)} |N_{G}(i)\cap N_{G}(j)| $, since a triangle $ <i,j,x,i> $ will be counted twice, one for each $j, x \in N_{G}(i)$. Hence, 
	\begin{eqnarray*}
		\text{Total number of triangles in the graph }&=& \dfrac{1}{3} \sum\limits_{i} NT(i)\\
		&=& \dfrac{1}{6} \sum\limits _{i}\sum\limits_{j\in N_{G}(i)}|N_{G}(i)\cap N_{G}(j)|\\
		&= &\dfrac{1}{6}\sum\limits_{i}\sum\limits_{j\in N_{G}(i)}\Bigg(|N_{G}(j)|-|N_{G}(j)-N_{G}(i)|\Bigg) \\
		&=&\dfrac{1}{6}\sum\limits_{i}\sum\limits_{j\in N_{G}(i)}\Bigg(|\eta_{jj}|-\eta_{ij}\Bigg)  
	\end{eqnarray*} 
	Hence the claim.
\end{proof}
\begin{remark}
	It is well known that number of triangle in a graph is equal to $\dfrac{1}{6}  Trace(A^3) $ or $ \dfrac{1}{6}\sum\limits_{i=1}^{n}\lambda_{i}^{3}  $ where $ A $ is the adjacency matrix of the graph and $ \lambda_{i}, 1\leq i\leq n $ is the eigenvalue of $ A $. Note that if we want to count a triangle using the $ \mathcal{NM}(G) $ the computational time involved is very less when compared to compute $\dfrac{1}{6}  Trace(A^3) $ or $ \dfrac{1}{6}\sum\limits_{i=1}^{n}\lambda_{i}^{3}  $
\end{remark}
\begin{proposition}
	Given a graph $ G $, the number of 4 cycles(including induced and non-induced) is equal to $ \frac{1}{4}\sum\limits_{i=1}^n\Bigg(\sum\limits_{j\in N_{G}(i)}{|\eta_{jj}|-\eta_{ij}\choose 2} +\sum\limits_{j\notin N_{G}(i)}{|\eta_{ij}|\choose 2}\Bigg) $
\end{proposition}
\begin{proof}
	Given a graph  $G$, the number of 4-cycles containing the vertex $i$ is given by $\sum \limits_{j=1,j\neq i}{|N_{G}(i)\cap N_{G}(j)|\choose 2}$.   
	Hence the total number of $4$-cycles (both induced and not induced) can be given by, 
	\begin{eqnarray*}
		\frac{1}{4}\sum \limits_{i=1}^n\sum \limits_{j=1,j\neq i}{|N_{G}(i)\cap N_{G}(j)|\choose 2} & = & \frac{1}{4}\sum \limits_{i=1}^n\Bigg(\sum\limits_{j\in N_{G}(i)}{|N_{G}(i)\cap N_{G}(j)|\choose 2}+\sum \limits_{j\notin N_{G}(i)}{|N_{G}(i)\cap N_{G}(j)|\choose 2}\Bigg)\\
		& = &\frac{1}{4}\sum\limits_{i=1}^n\Bigg(\sum\limits_{j\in N_{G}(i)}{|\eta_{jj}|-\eta_{ij}\choose 2} +\sum\limits_{j\notin N_{G}(i)}{|\eta_{ij}|\choose 2}\Bigg)\\
	\end{eqnarray*}
\end{proof}
\begin{remark}
	Note that in the above proof, $\dfrac{1}{4}\sum \limits_{i=1}^n\sum\limits_{j\notin N_{G}(i)}{|\eta_{ij}|\choose 2} $ gives a count of the total number of induced $C_{4}$ plus half the number of $ K_{4}-\{e\}. $\\
	Similarly, $  \dfrac{1}{4}\sum \limits_{i=1}^n\sum\limits_{j\in N_{G}(i)}{|\eta_{jj}|-\eta_{ij}\choose 2} $ gives the total number of $ K_{4} $ along with half the number of $ K_{4}-\{e\} $ in the graph.
\end{remark}
\begin{proposition}\label{Induced C4 free}
	A graph $ G $ is $ C_4$-free if and only if $ \eta_{ij}\geq -1, \forall (i,j)\notin E(G)$.
\end{proposition}
\begin{proof}
	By the definition of $ \mathcal{NM}(G) $, we can conclude
	\begin{eqnarray*}
		\eta_{ij}\geq -1&\iff & |N_{G}(i)\cap N_{G}(j)|\leq 1,(i,j)\notin E(G)\\
		&\iff & G \text{ has no induced } C_{4}
	\end{eqnarray*}
\end{proof}
Recall that the girth of a graph is the length of a shortest cycle contained in the graph.
\begin{proposition}
	A graph $ G $ has girth at least 5 if and only if $\eta_{ij}=|\eta_{jj}|,\forall (i,j)\in E(G)$ and $ \eta_{ij}\geq -1, \forall (i,j)\notin E(G)$
\end{proposition}
\begin{proof}
	By \textnormal{Proposition \ref{Triangle free} } we get, \begin{eqnarray*}
		\eta_{ij}=|\eta_{jj}|, \forall (i,j)\in E(G) &\iff & G  \text{ is Triangle free }. 
	\end{eqnarray*}
	and by \textnormal{Proposition \ref{Induced C4 free} } we get,
	\begin{eqnarray*}
		\eta_{ij}\geq -1&\iff & |N_{G}(i)\cap N_{G}(j)|\leq 1,(i,j)\notin E(G)\\
		&\iff & G \text{ has no induced } C_{4}
	\end{eqnarray*}
	Therefore we can conclude $ G $  has girth atleast 5.
\end{proof}

\section{Conclusion}

In this paper, we have introduced a new graph matrix ($\mathcal{NM}(G) $) that can be associated with a graph to reveal more information when compared to the adjacency matrix. We have also systematically demonstrated the equivalence of the $ \mathcal{NM}(G) $ and the product of two other existing graph matrices, namely adjacency and Laplacian matrices. Further, we have endorsed its relationship with the concept of level decomposition of graph. 

Further, we have also substantiated the usefulness of the $ \mathcal{NM}(G) $ by identifying numerous properties that can be revealed with the aid of this matrix. In this process, we have found that many simple properties such as counting the number of triangles in a graph can be done in no-time.  

As an extension of this current work, in a sequel, our subsequent research article comprises of the study of the $ \mathcal{NM} $ spectrum. We are also analyzing the algorithmic properties of this matrix and other interesting graph properties that can be revealed. 

\subsection{Future scope}
In our first attempt to analyze a new graph matrix, we have only studied its correctness and very few graph and matrix properties in this paper. This graph matrix seems to be quite promising and can be applicable in studying problems relating to domination in graphs and graph isomorphism problem. We have already initiated our study in this direction.

\section*{Acknowledgements}
The authors would like to acknowledge and thank DST-SERB Young Scientist Scheme, India [Grant No. SB/FTP/MS-050/2013] for their support to carry out this research at SRM Research Institute, SRM University. Mr. K. Sivakumar would also like to thank SRM Research Institute for their support during the preparation of this manuscript.  

%\section*{References} 


\begin{thebibliography}{9}
	\bibitem{GraphsMatrices} R.B. Bapat, Graphs and matrices, \textit{Springer}, (2014).   
	\bibitem{Bapatetal} R.B. Bapat, S. J. Kirkland, K.M. Prasath and S. Puntanen(Eds), {Combinatorial Matrix Theory and generalized inverses of matrices}, \textit{Springer India}, (2013).  
	\bibitem{Graphtheory} J.A. Bondy and U.S.R. Murty, {Graph Theory},\textit{ Graduate Texts in Mathematics, Springer}, (2008).   
	\bibitem{MCM3} D Janezic, A. Milicevic, S.Nikolic and N.Trinajstic, {Graph-Theoretical Matrices in Chemistry},\textit{ Mathematical Chemistry Monographs}, \textbf{3 }(2007).  		 
	\bibitem{SURVEY} Russell Merris, {The Laplacian matrices of graphs: a survey}, \textit{Linear Algebra and its Applications},\textbf{(197-198) }(1994), 143-176.
	\bibitem{BIOLOGY} Georgios A Pavlopoulos, Maria Secrier, Charalampos N Moschopoulos, Theodoros G Soldatos, Sophia Kossida, Jan Aerts, Reinhard Schneider and Pantelis G Bagos, {Using graph theory to analyze biological networks}, \textit{BioData Mining}, \textbf{4:10 }(2011), DOI: 10.1186/1756-0381-4-10.
	\bibitem{GraphTheoryWest} D. B. West, Introduction to Graph Theory, \textit{Pearson}, \textbf{2 }(2000).  
\end{thebibliography}
\end{document}